\documentclass[a4paper,UKenglish]{lipics-v2016}
 
\usepackage{microtype}
\usepackage{amsmath,array,bm}
\usepackage{amsmath,amsfonts}
\usepackage{amsthm}
\usepackage{amssymb,latexsym}
\usepackage{algorithm}
\usepackage{algorithmicx}
\usepackage[]{algpseudocode}
\usepackage[dvipsnames]{xcolor}
\usepackage{graphicx}
\usepackage{caption}
\usepackage{enumitem}
\usepackage{float}
\usepackage{hyperref}
\usepackage{cleveref}
\usepackage[square,sort,comma,numbers]{natbib}
\usepackage{tikz}
\usepackage{xspace}
\usepackage{xcolor}
\usepackage{mdframed}
\usepackage{multirow}
\usepackage{hhline}

\newtheorem{claim}{Claim}
\theoremstyle{definition}

\crefname{invar}{invariant}{invariants}
\crefname{constr}{constraint}{constraints}
\crefname{tbl}{table}{tables}
\crefname{lem}{lemma}{lemmas}

\title{On Partial Covering For Geometric Set Systems}
\titlerunning{On Partial Covering For Geometric Set Systems} 

\author[1]{Tanmay Inamdar}
\author[2]{Kasturi Varadarajan}
\affil[1]{Department of Computer Science, University of Iowa, Iowa City, IA, USA.\\
  \texttt{tanmay-inamdar@uiowa.edu}}
\affil[2]{Department of Computer Science, University of Iowa, Iowa City, IA, USA.\\
  \texttt{kasturi-varadarajan@uiowa.edu}}
\authorrunning{T.\, Inamdar and K.\ Varadarajan} 

\Copyright{Tanmay Inamdar and Kasturi Varadarajan}

\subjclass{}
\keywords{}

\newcommand{\real}{\mathbb{R}}

\newcommand{\X}{X}
\newcommand{\R}{\mathcal{R}}
\newcommand{\T}{\mathcal{T}}
\renewcommand{\S}{\mathcal{S}}
\renewcommand{\cdots}{\ldots}
\newcommand{\lr}[1]{\left(#1\right)}
\newcommand{\Red}[1]{{\color{red} #1}}
\newcommand{\Xc}{\X_c}
\newcommand{\Sc}{\S_c}
\renewcommand{\epsilon}{\varepsilon}
\newcommand{\SC}{\textsf{SC}\xspace}

\newcommand{\PSC}{\textsf{PSC}\xspace}
\newcommand{\LP}{\textsf{LP}\xspace}
\DOIPrefix{}

\begin{document}

\maketitle

\begin{abstract}
We study a generalization of the Set Cover problem called the \emph{Partial Set Cover} in the context of geometric set systems. The input to this problem is a set system $(\X, \S)$, where $\X$ is a set of elements and $\S$ is a collection of subsets of $\X$, and an integer $k \le |\X|$. The goal is to cover at least $k$ elements of $\X$ by using a minimum-weight collection of sets from $\S$. The main result of this article is an \LP rounding scheme which shows that the integrality gap of the Partial Set Cover \LP is at most a constant times that of the Set Cover \LP for a certain projection of the set system $(\X, \S)$. As a corollary of this result, we get improved approximation guarantees for the Partial Set Cover problem for a large class of geometric set systems.
\end{abstract}

\section{Introduction}
In the Set Cover (\SC) problem, the input is a set system $(\X, \S)$, where $\X$ is a set of $n$ elements, and $\S$ is a collection of subsets of $\X$ . The goal is to find a minimum size collection $\S' \subseteq \S$ that \emph{covers} $\X$, i.e., the union of the sets in $\S'$ contains the elements of $\X$. In the weighted version, each set $S_i \in \S$ has a non-negative weight $w_i$ associated with it, and we seek to minimize the weight of $\S'$. A simple greedy algorithm finds a solution that is guaranteed to be within $O(\log n)$ factor from the optimal, and it is not possible to do better in general, under certain standard complexity theoretic assumptions \cite{Feige1998}.

The question of whether we can improve the $O(\log n)$ bound has been extensively studied for geometric set systems. We focus on three important classes -- covering, hitting, and art gallery problems. In the Geometric Set Cover problem, $\X$ typically consists of points in $\real^d$, and $\S$ contains sets induced by a certain class of geometric objects via containment. For example, each set in $\S$ might be the subset of $\X$ contained in a hypercube. Some of the well-studied examples include covering points by disks in plane, fat triangles, etc. In the Geometric Hitting Set problem, $\X$ is a set of geometric objects, and each set in $\S$ is the subset consisting of all objects in $\X$ that are pierced by some point. In an example of the art gallery problem, $\X$ consists of a set of points in a simple polygon, and each set in $\S$ is the subset consisting of all points in $\X$ that can be seen by some vertex of the polygon \cite{KingKirk}. Thus, the set system here is defined by visibility. 

For many such geometric set systems, it is possible to obtain approximation guarantees better 
than $O(\log n)$. We survey two of the main approaches to obtain such guarantees. The first and the most successful approach is based on the \SC Linear Program (\LP) and its connection to $\epsilon$-nets. For completeness, we state the standard \SC \LP for the weighted case.
	\begin{alignat}{3}
	\text{minimize}   \displaystyle&\sum\limits_{S_i \in \S} w_{i}x_{i} & \nonumber\\
	\text{subject to} \displaystyle&\sum\limits_{i:e_{j} \in S_{i}}   x_{i} \geq 1,  \quad & \quad e_j \in \X \\
	\displaystyle &\qquad \quad x_i \ge 0, & \quad S_i \in \S
	\end{alignat}
For the unweighted case, \citet{Even2005} showed that if for a certain set system, $O\lr{\frac{1}{\epsilon} \cdot g\lr{\frac{1}{\epsilon}}}$ size $\epsilon$-nets exist, then the integrality gap of the \SC \LP is $O(g(OPT))$, where $OPT$ is the size of the optimal solution. This result is constructive, in that an efficient algorithm for constructing $\epsilon$-nets also yields an efficient algorithm for obtaining an $O(g(OPT))$ approximation. (A similar result was obtained earlier by \citet{BronnimannG1995}, without using LP machinery). It is fairly well-known (\cite{Clarkson1987, HausslerV1987}) that for a large class of geometric set systems, $\epsilon$-nets of size $O\lr{\frac{1}{\epsilon} \log\lr{\frac{1}{\epsilon}}}$ can be computed efficiently, which implies $O(\log(OPT))$ approximation for the set cover problem on the corresponding geometric set system. \citet{ClarksonV2007} showed that if the \emph{union complexity} of any set of $n$ objects is $O(n\cdot h(n))$, then $\epsilon$-nets of size $O\lr{\frac{1}{\epsilon} \cdot h\lr{\frac{1}{\epsilon}}}$ exist. \citet{AronovES2010} gave a tighter bound of $O(\frac{1}{\epsilon} \cdot \log h(\frac{1}{\epsilon}))$ on the size of $\epsilon$-nets for the objects of union complexity $O(n\cdot h(n))$ (see also \cite{VaradarajanUnion2009}). Some of these results were extended to the weighted case in \cite{VaradarajanWGSC2010, ChanGKS12} by a technique called \emph{quasi-uniform sampling}. We summarize some of these $\epsilon$-net based results for the set cover problem for geometric set systems in the accompanying table.

\begin{table}[hbt] \label[tbl]{tbl:wsc-results}
	\centering
	\begin{tabular}{|c|c|c|}
		\hline
		$\X$ & Geometric objects inducing $\S$ & Integrality Gap of \SC \LP\\
		\hline
		\multirow{2}{*}{Point in $\real^2$} & Disks (via containment) & $O(1)$ \\
		\hhline{|~|--}
		& Fat triangles (containment) & $O(\log \log^* n)$ \\
		\hline
		\multirow{2}{*}{Points in $\real^3$} & Unit cubes (containment) & $O(1)$ \\
		\hhline{|~|--}
		& Halfspaces (containment) & $O(1)$ \\
		\hline
                Rectangles in $\real^3$ & Points (via piercing) & $O(\log \log n)$ \\
                \hline
                Points on 1.5D terrain & Points on terrain (via visibility) & $O(1)$ \cite{ElbTerrain} \\
                \hline 
	\end{tabular}
\caption{LP-based approximation ratios for \SC. See \cite{ClarksonV2007,AronovES2010,VaradarajanWGSC2010, EASFat, ChanGKS12} for the references establishing these bounds. Except for piercing rectangles in $\real^3$ by points, these bounds hold for the weighted \SC. For these problems, we obtain analogous results for weighted \PSC.}
\end{table}
Another approach for tackling \SC for geometric set systems is by combinatorial algorithms. The dominant paradigm from this class is the simple Local Search algorithm. The effectiveness of Local Search was first demonstrated by \citet{MustafaR2010}, who gave the first PTAS for covering points by disks in plane. There have been a series of results that build on their work, culminating in \citet{GovindarajanRRB2016}, who show that Local Search yields a PTAS for \SC for a fairly general class of objects such as pseudodisks and non-piercing regions in plane. \citet{KrohnTerrain} gave a PTAS for the terrain guarding problem, where the geometric set system is defined by visibility. Another common strategy, called the \emph{shifting strategy}, was introduced by \citet{HochbaumM1985}. They give a PTAS for covering points by unit balls in $\real^d$; however in this case the set $\S$ consists of \emph{all} unit balls in $\real^d$. \citet{ChanPiercing} gave a PTAS for piercing a set of fat objects in $\real^d$ using a minimum number of points from $\real^d$.

Now we turn to the Partial Set Cover (\PSC) problem. The input to \PSC is the same as that to the \SC, along with an additional integer parameter $k \le |\X|$. Here the goal is to cover at least $k$ elements from $\X$ while minimizing the size (or weight) of the solution $\S' \subseteq S$. It is easy to see that \PSC is a generalization of \SC, and hence it is at least as hard as \SC. We note here that another classical problem that is related to both of these problems is the so-called Maximum Coverage (\textsf{MC}) problem. In this problem, we have an upper bound on the number of sets that can be chosen in the solution, and the goal is to cover the maximum number of elements. It is a simple exercise to see that an exact algorithm for the unweighted \PSC can be used to solve \textsf{MC} exactly, and vice versa. However the reductions are not approximation-preserving. In particular, the greedy algorithm achieves $1 - 1/e$ approximation guarantee for \textsf{MC} --- which is essentially the best possible --- whereas it is \textsf{NP}-hard to approximate \PSC within $o(\log n)$ factor in general. We refer the reader to \cite{KhullerMN1999} for a generalization of \textsf{MC} and a survey of results.

For \PSC, the greedy algorithm is shown to be an $O(\log \Delta)$ approximation in \cite{Slavik1997}, where $\Delta$ is the size of the largest set in $\S$. \citet{Gandhi2004} give a primal-dual based algorithm which achieves an approximation guarantee of $f$, where $f$ is the maximum frequency of any element in the sets. A special case of \PSC is the Partial Vertex Cover (\textsf{PVC}) problem, where we need to pick a minimum size (or weight) subset of vertices that covers at least $k$ edges of the graph. \citet{BshoutyL1998} describe a $2$-approximation based on \LP rounding for \textsf{PVC}. 
Improvements for some special classes of graphs are described in \citet{Gandhi2004}. See also \cite{Mestre2009, KonemannPS2011} for more recent results on \textsf{PVC}, \PSC, and related problems.

While \SC for various geometric set systems has been studied extensively, there are relatively fewer works studying \PSC in the geometric setting. \citet{Gandhi2004} give a PTAS for a geometric version of \PSC where $\S$ consists of \emph{all} unit disks in the plane. They provide a dynamic program on top of the standard shifting strategy of \citet{HochbaumM1985}, thus adapting it for \PSC. Using a similar technique, \citet{GlasserRS2008} give a PTAS for a generalization of partial geometric covering, under a certain assumption on the \emph{density} of the given disks.  \citet{ChanN2015} give a PTAS for covering points by unit squares in the plane. 


\paragraph*{Our Results and Techniques}
Suppose that we are given a \PSC instance $(\X, \S, k)$. For any set of elements $\X_1 \subseteq \X$, let $\S_{\X_1} := \{S \cap \X_1 \mid S \in \S \}$ denote the projected set system. Suppose also that for any projected \SC instance $(\X_1, \S_{|\X_1})$, (where $\X_1 \subseteq \X$) and a corresponding feasible \SC \LP solution $\sigma_1$, we can round $\sigma_1$ to a feasible integral \SC solution with cost at most $\beta$ times that of $\sigma_1$. That is, we suppose that we can efficiently compute a $\beta$-approximation for the $\SC$ instance $(\X_1, \S_{|\X_1})$ by solving the natural LP relaxation and rounding it. Then, we show that we can round the natural \PSC \LP to an integral solution to within a $2\beta + 2$ factor. By the previous discussion about existence of such rounding algorithms for \SC \LP for a large class of geometric objects (cf. \Cref{tbl:wsc-results}), we get the same guarantees for the corresponding \PSC instances as well (up to a constant factor). For clarity, we describe a sample of these applications.

\begin{enumerate}
\item Suppose we are given a set $P$ of $n$ points and a set $\T$ of {\em fat} triangles in the plane and a positive weight for each triangle in $\T$. We wish choose a subset $\T' \subseteq 
\T$ of triangles that covers $P$, and minimize the weight of $\T'$, defined to be the sum of the weights of the triangles in it. This is a special case of weighted \SC obtained by setting $\X = P$, and adding the set $T \cap P$ to $\S$ for each triangle in $T \in \T$, with the same weight. There is an $O(\log \log^* n)$ approximation for this problem based on rounding \SC \LP \cite{EASFat, ChanGKS12}. We obtain the same approximation guarantee for the partial covering version, where we want a minimum weight subset of $\T$ covering any $k$ of the points in $P$.
\item Suppose we are given a set $\R$ of $n$ axis-parallel rectangles and a set $P$ of points in $\real^3$, and we wish to find a minimum cardinality subset of $P$ that hits (or pierces) $\R$. This a special case of \SC obtained by setting $\X = \R$, and adding the set $\{R \in \R \ | \ p \in R\}$ to $\S$ for each point $p \in P$. There is an $O(\log \log n)$ approximation for this problem based on rounding \SC \LP \cite{AronovES2010}. Thus, we obtain the same approximation guarantee for the partial version, where we want a minimum cardinality subset of $P$ piercing any $k$ of the rectangles in $\R$.
\item Suppose we have a 1.5D terrain (i.e.,an $x$-monotone polygonal chain in $\real^2$), a set 
$P$ of points and a set $G$ of $n$ points, called guards, on the terrain along with a positive weight for each guard in $G$. The goal is to choose a subset $G' \subset G$ such that each point in $P$ is seen by some guard in $G$, and minimize the weight of $G'$. Two points $p$ and $g$ on the terrain see each other if the line segment connecting them does not contain a point below the terrain. This is a special case of \SC obtained by setting $\X = P$, and adding the set $\{p \in P \ | \ g \mbox{ sees } p\}$ to $\S$ for each guard $g \in G$. There is an $O(1)$ approximation guarantee for this problem based on rounding \SC \LP \cite{ElbTerrain}. Thus, we obtain an $O(1)$ approximation for the partial version, where we want a minimum weight subset of $G$ that sees any $k$ of the points in $P$.
\end{enumerate}     

Our algorithm for rounding a solution to the natural \PSC \LP corresponding to partial cover 
instance $(\X, \S, k)$ proceeds as follows. Let $\X_1$ be the elements that are covered by the 
LP solution to an extent of at least $1/2$. By scaling the \LP solution by a factor of $2$, we get
a feasible solution to the \SC \LP corresponding to $(\X_1, \S_{|X_1})$, which we round using the 
\LP-based $\beta$-approximation algorithm. For the set $\X \setminus \X_1$, the \LP solution 
provides a total fractional coverage of at least $k - |\X_1|$. Crucially, each element of
$\X \setminus \X_1$ is {\em shallow} in that it is covered to an extent of at most $1/2$. We use this observation to round the \LP solution to an integer solution, of at most twice the cost, that 
covers at least $k - |\X_1|$ points of $\X \setminus \X_1$.  This rounding step and its analysis are inspired by the \textsf{PVC} rounding scheme of \cite{BshoutyL1998}, however there are certain 
subtleties in adapting it to the \PSC problem. To the best of our knowledge, this connection between the \SC \LP and \PSC \LP was not observed before.

The rest of this article is organized as follows. In \Cref{sec:prelim}, we describe the standard \LP formulation for the \PSC problem, and give an integrality gap example. We describe how to circumvent this integrality gap by preprocessing the input in \Cref{sec:preproc}. Finally, in \Cref{sec:rounding}, we describe and analyze the main \LP rounding algorithm.

\section{Preliminaries} \label{sec:prelim}
\textbf{\LP Formulation}

We use the following Integer Programming formulation of \PSC (left). Here, for each element $e_j \in \X$, the variable $z_j$ denotes whether it is one of the $k$ elements that are chosen by the solution. For each such chosen element $e_j$, the first constraint ensures that at least one set containing it must be chosen. The second constraint ensures that at least $k$ elements are covered by the solution. We relax the integrality \cref{constr:integral-z,constr:integral-x}, and formulate it as a Linear Program (right).

\begin{minipage}{0.475\textwidth}
		\begin{mdframed}[backgroundcolor=gray!9] 
			\fontsize{8.5}{9.5}\selectfont
			\begin{alignat}{3}
			\text{minimize}   \displaystyle&\sum\limits_{S_i \in \S} w_{i}x_{i} & \nonumber\\
			\text{subject to} \displaystyle&\sum\limits_{i:e_{j} \in S_{i}}   x_{i} \geq z_j,  \quad &e_j \in \X \nonumber\\
			\displaystyle&\sum_{e_j \in \X}z_j \ge k, & \nonumber\\
			\displaystyle &z_j \in \{0, 1\}, & e_j \in \X \label[constr]{constr:integral-z}\\
			\displaystyle &x_i \in \{0, 1\}, & S_i \in \S \label[constr]{constr:integral-x}
			\end{alignat}
			\centering Integer Program
	\end{mdframed}
\end{minipage}
\begin{minipage}{0.48\textwidth}
			
	\begin{mdframed}[backgroundcolor=gray!9] 
		\fontsize{8.5}{9.5}\selectfont
		\begin{alignat}{3}
			\text{minimize}   \displaystyle&\sum\limits_{S_i \in \S} w_{i}x_{i} & \label[constr]{constr:objective}\\
			\text{subject to} \displaystyle&\sum\limits_{i:e_{j} \in S_{i}}   x_{i} \geq z_j,  \quad &e_j \in \X \label[constr]{constr:coverage}\\
			\displaystyle&\sum_{e_j \in \X}z_j \ge k, & \label[constr]{constr:k-cover}\\
			\displaystyle &z_j \in [0, 1], & e_j \in \X \label[constr]{constr:fractional-z}\\
			\displaystyle &x_i \in [0, 1], & S_i \in \S \label[constr]{constr:fractional-x}
		\end{alignat}
		\centering Linear Program
	\end{mdframed}
\end{minipage}
\\

Since \SC is a special case of \PSC where $k = n$, the corresponding \LP can be obtained by setting $k$ appropriately in \Cref{constr:k-cover}. However, in this case, the \LP can be further simplified as described earlier. We denote the cost of a \PSC \LP solution $\sigma = (x, z)$, for the instance $(\X, \S)$, as $cost(\sigma) := \sum_{S_i \in \S} w_i x_i$, and the cost of an \SC \LP solution is defined in exactly the same way. Also, for any collection of sets $\S' \subseteq \S$, we define $w(S') := \sum_{S_i \in \S'} w_i$. Finally, for a \PSC instance $(\X, \S, k)$, let $OPT(\X, \S, k)$ denote the cost of an optimal solution for that instance.

Unlike \SC \LP, the integrality gap of \PSC \LP can be as large as $O(n)$, even for the unweighted case. 

\textbf{Integrality Gap}. Consider the set system $(\X, \S)$, where $\X = \{e_1, \cdots, e_n\}$, and $\S = \{S_1\}$, where $S_1 = X$. Here, $k = 1$, so at least one element has to be covered. The size of the optimal solution is $1$, because the only set $S_1$ has to be chosen. However, consider the following fractional solution $\sigma = (x, z)$, where $z_j = \frac{1}{n}$ for all $e_j \in \X$, and $x_1 = \frac{1}{n}$, which has the cost of $\frac{1}{n}$. This shows the integrality gap of $n$.

However, \citet{Gandhi2004} show that after ``guessing'' the heaviest set in the optimal solution, the integrality gap of the \LP corresponding to the residual instance is at most $f$, where $f$ is the maximum frequency of any element in the set system (this also follows from the modification to the rounding algorithm of \cite{BshoutyL1998}, as commented earlier). In this article, we show that after guessing the heaviest set in the optimal solution, the residual instance has integrality gap at most $2\beta+2$, where $\beta$ is the integrality gap of the \SC \LP for some projection of the same set system.

\section{Preprocessing} \label{sec:preproc}
Let $(\X', \S', k')$ be the original instance. To circumvent the integrality gap, we preprocess the given instance to ``guess'' the heaviest set in the optimal solution, and solve the residual instance as in \cite{BshoutyL1998, Gandhi2004} -- see \Cref{alg:partialcover}. Let us renumber the sets $\S' = \{S_1, \ldots, S_m\}$, such that $w_1 \le w_2 \le \ldots \le w_m$. For each $S_i \in \S'$, let $\S_i = \{S_1, S_2, \ldots, S_{i-1}\}$, and $\X_i = \X' \setminus S_i$. We find the approximate solution $\Sigma_i$ for this residual instance $(\X_i, \S_i, k_i)$ with coverage requirement $k_i = k-|S_i|$, if it is feasible (i.e.  $\left|\bigcup_{S \in \S_i} S \cap X_i \right| \ge k_i$). We return $\Sigma = \arg\min_{S_i \in \S'} w(\Sigma_i \cup \{S_i\})$ over all $S_i$ such that the residual instance $(\X_i, \S_i, k_i)$ is feasible.

\begin{algorithm}[hbt]
	\caption{PartialCover$(\X', \S', k')$} \label{alg:partialcover}
	\begin{algorithmic}[1] 	
		\State Sort and renumber the sets in $\S' = \{S_1, \cdots, S_m \}$ such that $w_1 \le \ldots \le w_m$.
		\For{$i = 1$ to $m$}
			\State $\S_i \gets \{S_1, \ldots, S_{i-1}\}$
			\State $\X_i \gets \X' \setminus S_i$
			\State $k_i \gets k' - |S_i|$
			\If{$(\S_i, \X_i, k_i)$ is feasible}
				\State $\Sigma_i \gets $ approximate solution to $(\X_i, \S_i, k_i)$
			\Else
				\State $\Sigma_i \gets \perp$
			\EndIf
		\EndFor
		\State \Return $\arg\min_{S_i \in \S': \Sigma_i \neq \perp} w(\Sigma_i \cup \{S_i\})$
	\end{algorithmic}
\end{algorithm}

\begin{lemma} \label[lem]{lem:lemma-partial-cover}
	Let $\Sigma^*$ be the optimal partial cover for the instance $(\X', \S', k')$, and let $S_p$ be the heaviest set in $\Sigma^*$. Let $\Sigma_p$ be the approximate solution to $(\X_p, \S_p, k_p)$ returned by the Rounding Algorithm of \Cref{thm:rounding-theorem}, and $\Sigma'$ be the solution returned by \Cref{alg:partialcover}. Then,
	\begin{enumerate}
		\item $OPT(\X', \S', k') = OPT(\X_p, \S_p, k_p) + w_p$
		\item $w(\Sigma') \le w(\Sigma_p \cup \{S_p\}) \le (2\beta + 2) \cdot OPT(\X', \S', k')$
	\end{enumerate}
\end{lemma}
\begin{proof}
	Since the optimal solution $\Sigma^*$ contains $S_p$, $\Sigma^*_p := \Sigma^* \setminus \{S_p\}$ covers at least $k' - |S_p| = k_p$ elements from $\X' \setminus S_p$. Therefore, $\Sigma^*_p$ is feasible for $(\X_p, \S_p, k_p)$. If we show that $w(\Sigma^*_p) = OPT(\X_p, \S_p, k_p)$, the first part follows. Assume for contradiction that there is a set $\Sigma'_p \subseteq \S_p$ such that $w(\Sigma'_p) = OPT(\X_p, \S_p, k_p) < w(\Sigma^*_p)$. However, $\Sigma'_p$ covers at least $k_p = k' - |S_p|$ elements from $\X' \setminus S_p$. So $\Sigma'_p \cup \{S_p\}$ covers at least $k'$ elements from $\X'$, and has weight $w(\Sigma_p') + w_p < w(\Sigma^*_p) + w_p = w(\Sigma^*)$, which is a contradiction.
	
	From \Cref{thm:rounding-theorem}, we have an approximate solution $\Sigma_p$ to the instance $(\X_p, \S_p, k_p)$ such that $w(\Sigma_p) \le (2\beta + 2) \cdot OPT(\X_p, \S_p, k_p) + B$, where $B = w_p$ is the weight of the heaviest set in the optimal solution. Now \Cref{alg:partialcover} returns a solution whose cost is at most $w(\Sigma_p \cup \{S_p\}) \le (2\beta + 2) \cdot OPT(\X_p, \S_p, k_p) + w_p + w_p \le (2\beta+2) \cdot (OPT(\X_p, \S_p, k_p) + w_p) \le (2 \beta+2) \cdot OPT(\X', \S', k').$ We use the result from part 1 in the final inequality.
\end{proof}

We summarize our main result in the following theorem, which follows easily from \Cref{lem:lemma-partial-cover}.

\begin{theorem}
	Let $(\X', \S')$ be a set system, such that we can round a feasible \SC \LP for any projected set system $(\X_1, \S'_{|\X_1})$ to within $\beta$ factor, where $\X_1 \subseteq \X'$. Then, we can find a $(2\beta+2)$-factor approximation for the partial set cover instance $(\X', \S', k')$, where $1 \le k' \le n$.
\end{theorem}

\section{Rounding Algorithm} \label{sec:rounding}
Suppose that we have guessed the maximum weight set $S_p \in \S'$ in the optimal solution for the original instance $(\X', \S', k')$, as described in the previous section. Thus, we now have the residual instance $(\X_p, \S_p, k_p)$, where $\X_p = (\X' \setminus S_p), \S_p = \{S_1, S_2, \cdots, S_{p-1}\}$, and $k_p = k' - |S_p|$. We solve the \LP corresponding to the \PSC instance $(\X_p, \S_p, k_p)$ to obtain an optimal \LP solution $\sigma^* = (x, z)$. In the following, we describe a polynomial time algorithm to round \PSC \LP on this instance.

Let $0 < \alpha \le 1/2$ be a parameter (finally we will set $\alpha = 1/2$). Let $\X_1 = \{e_j \in \X_p \mid \sum_{i: e_j \in S_i} x_i \ge \alpha\}$ be the set of elements that are covered to an extent of at least $\alpha$ by the \LP solution.

We create an instance $\sigma_1$ of a feasible set cover \LP for the instance $(\X_1, {\S_p}_{|\X_1})$ as follows. For all sets $S_i \in \S_p$, we set $x'_i = \min\{\frac{x_i}{\alpha}, 1\}$. Note that cost of this fractional solution is at most $\frac{1}{\alpha}$ times that of $\sigma^*$. Also, note that $\sigma_1$ is feasible for the \SC \LP because for any element $e_j \in \X_1$, we have that 
$$\sum_{i: e_j \in S_i} x'_i = \sum_{i: e_j \in S_i} \min\left\{1, \frac{x_i}{\alpha}\right\} \ge \min\bigg\{ 1, \frac{1}{\alpha} \sum_{i:e_j \in S_i} x_i\bigg\} \ge 1$$

Suppose that there exists an efficient rounding procedure to round a feasible \SC \LP solution $\sigma_1$, for the instance $(\X_1, {\S_p}_{|\X_1})$ to a solution with weight at most $\beta \cdot cost(\sigma_1)$. In the remainder of this section, we describe an algorithm (\Cref{alg:rounding}) for rounding $\sigma^* = (x, z)$ into a solution that (1) covers at least $k_p - |\X_1|$ elements from $\X_p \setminus \X_1$, and (2) has cost at most $\frac{1}{\alpha} \cdot cost(\sigma^*) + B$, where $B$ is the weight of the heaviest set in $\S_p$. Combining the two solutions thus acquired, we get the following theorem.

\begin{theorem} \label{thm:rounding-theorem}
	There exists a rounding algorithm to round a partial cover \LP corresponding to $(\X_p, \S_p, k_p)$, which returns a solution $\Sigma_p$ such that $w(\Sigma_p) \le (2\beta+2) \cdot OPT(\X_p, \S_p, k_p) + B$, where $B$ is the weight of the heaviest set in $\S_p$. 
\end{theorem}
\begin{proof}
	Let $\Sigma_p = \Sigma_{p1} \cup \Sigma_{p2}$, where $\Sigma_{p1}$ is the solution obtained by rounding $\sigma_1$, and $\Sigma_{p2} = \Sigma \cup \S_e$ is the solution returned by \Cref{alg:rounding}. By assumption, $\Sigma_{p1}$ covers $\X_1$, and $\Sigma_{p2}$ covers at least $k_p - |\X_1|$ elements from $\X_p \setminus \X_1$ by \Cref{lem:atleast-k}. Therefore, $\Sigma_p$ covers at least $k_p$ elements from $\X_p$.
	
	By assumption, we have that $w(\Sigma_{p1}) \le \beta \cdot cost(\sigma_1) \le \frac{\beta}{\alpha} cost(\sigma^*)$. Also, from \Cref{lem:mainlemma}, we have that $w(\Sigma_{p2}) \le \frac{1}{\alpha } cost(\sigma^*) + B$. We get the claimed result by combining previous two inequalities, setting $\alpha = 1/2$, and noting that $cost(\sigma^*) \le OPT(\X_p, \S_p, k_p)$.
\end{proof}

Let $\S_1 = \{S_i \in \S_p \mid x_i \ge \alpha\}$ be the sets that are opened to more than $\alpha$. Note that without loss of generality, we can assume that $\cup_{S_i \in \S_1} S_i \subseteq \X_1$. If $|\X_1| \ge k_p$, we are done. Otherwise, let $\X \gets \X_p \setminus \X_1$, $\S \gets \S_p \setminus \S_1$, and $k \gets k_p - |\X_1|$. Let $\sigma = (x, z)$ be the \LP solution $\sigma^*$ restricted to the instance $(\X, \S, k)$, that is, $x = (x_i \mid S_i \in \S ), z = (z_j \mid e_j \in \X )$. We show how to round $\sigma$ on the instance $(\X, \S, k)$ to find a collection of sets that covers at least $k$ elements from $\X$. In the following lemma, we show that the \LP solution $\sigma$ is feasible for the instance $(\X, \S, k)$.

\begin{lemma} \label[lem]{lem:feasibility}
	The \LP solution $\sigma = (x, z)$ is feasible for the instance $(\X, \S, k)$. Furthermore, $cost(\sigma) \le cost(\sigma^*)$.
\end{lemma}
\begin{proof}
	Note that $x_i$ and $z_j$ values are unchanged from the optimal solution $\sigma^*$, therefore the \cref{constr:fractional-x,constr:fractional-z} are satisfied.
	
	 Note that by definition, for any element $e_j \in \X$, $e_j \not\in \cup_{S_{i'} \in \S_1} S_{i'}$, and $e_j \not\in \X_1$. Therefore, by \Cref{constr:coverage}, we have that $\displaystyle \sum_{i:e_j \in S_i} x_j = \sum_{i:e_j \in S_i, S_i \in \S} x_j \ge z_j$. 
	
	As for \Cref{constr:k-cover}, note that 
	\begin{alignat*}{1}
		\sum_{e_j \in \X_p} z_j &\ge k_p \tag{By feasibility of optimal solution $\sigma^*$}
		\\\implies \sum_{e_j \in \X} z_j &\ge k_p - \sum_{e_j \in \X_1} z_j \tag{$\X = \X_p \setminus \X_1$}
		\\\implies \sum_{e_j \in \X} z_j &\ge k_p - |\X_1| \tag{$z_j \le 1$ for $e_j \in \X_1$ by feasibility}
		\\\implies \sum_{e_j \in \X} z_j &\ge k \tag{$k = k_p - |\X_1|$}
	\end{alignat*}
	
	Finally, note that $cost(\sigma) = \sum_{S_i \in \S} w_i x_i \le \sum_{S_i \in \S_p} w_i x_i = cost(\sigma^*)$, because $\S \subseteq \S_p$, and the $x_i$ values are unchanged.
\end{proof}
\subsection{Algorithm for Rounding Shallow Elements}

\ifdefined\DEBUG

Now we describe \Cref{alg:rounding}. Recall that we now have a residual instance $(\X, \S, k)$, and a feasible \LP solution $\sigma = (x, z)$, such that i) For all sets $S_i \in \S$, $x_i \le \alpha$, and ii) For all elements $e_j \in \X, \sum_{i: e_j \in S_i} x_i \le \alpha$. The goal of the algorithm is to find a solution from $\S$, such that it covers at least $k$ elements, and its cost is not too large as compared to the cost of the \LP solution.

To this end, we pair up sets in each iteration, and try to increase the $x_i$ value of one of the sets at the expense of other. We maintain the feasibility and cost of the \LP solution in this process. This is carefully ensured in the procedure {\sc{RoundTwoSets}}.

Initially, we start with copies of the elements and the sets -- $\X_c$ and $\S_c$ respectively. We remove a set $S_i$ from $\S_c$ if its $x_i$ value becomes $0$ or $\alpha$. In the latter case when an $x_i$ value reaches $\alpha$, we remove all the elements it contains from $\X_c$. The sets which are removed from $\S_c$ by virtue of their $x_i$ value becoming $\alpha$ form the output of the rounding algorithm.

In each iteration, we select a set $S_a \in \S_c$ arbitrarily, and pair it up with another arbitrary set $S_b \in \S_c$, and pair it up with $S_a$. Depending on the ``cost-effectiveness'' of the sets $S_a$ and $S_b$ -- $\frac{|\X_c \cap S_a|}{w_a}$ and $\frac{|\X_c \cap S_b|}{w_b}$ respectively -- we call {\sc{RoundTwoSets}}($S_a, S_b, w, \sigma, \Xc, \Sc$), or {\sc{RoundTwoSets}}($S_b, S_a, w, \sigma, \Xc, \Sc$). In the former case, $x_a$ is increased at the expense of $x_b$, and in the latter case $x_b$ is increased at the expense of $x_a$.

Notice that if we paired up sets $S_a$ and $S_b$ arbitrarily and rounded using {\sc{RoundTwoSets}}, the feasibility of the \LP may not be maintained. In particular, we cannot ensure that for all elements $e_j \in \X_c$, $z_j \le 1$. To this end, once we choose a set $S_a \in \S_c$ arbitrarily, we fix it to be one of the paired sets until it is removed from $\S_c$ (because of its $x_a$ value becoming $0$ or $\alpha$), or it is the only remaining set in $\S_c$. By doing this, we maintain the following two invariants: 
\begin{enumerate}
	\item Let $\X_o = \{e_j \in \X_c \mid z_j \ge \alpha \}$. During the execution of while loop of \Cref{lin:outer-whileloop}, there exists a set $S_a \in \S_c$ such that $\X_o \subseteq S_c$. \label[invar]{invar:inv1}
	\item For all sets $S_i \in \S_c \setminus \{S_a\}$, the $x_i$ values are unchanged, unless and until the set $S_i$ is paired up with $S_a$. \label[invar]{invar:inv2}
\end{enumerate}
In \Cref{lem:z-feasibility}, we show that these invariants imply that the constraint is maintained. 

The invariants are trivially true before the start of the while loop. 

\fi

We have an \LP solution $\sigma$ for the \PSC instance $(\X, \S, k)$. Note that for any $S_i \in \S$, $x_i < \alpha$, and for any $e_j \in \X, \alpha > \sum_{i: e_j \in S_i} x_i \ge z_j$, i.e. each element is \emph{shallow}. We now describe \Cref{alg:rounding}, which rounds $\sigma$ to an integral solution to the instance $(\X, \S, k)$. At the beginning of \Cref{alg:rounding}, we initialize $\S_c$, the collection of ``unresolved'' sets, to be $\S$; and $\X_c$, the set of ``uncovered'' elements, to be $\X$.

At the heart of the rounding algorithm is the procedure {\sc{RoundTwoSets}}, which takes input two sets $S_1, S_2 \in \S_c$, and rounds the corresponding variables $x_1, x_2$ such that either $x_1$ is increased to $\alpha$, or $x_2$ is decreased to $0$ (cf. \Cref{lem:weight-maintained} part 3). A set is removed from $\S_c$ if either of these conditions is met. In addition, if $x_i$ reaches $\alpha$, then the set $S_i$ is added to $\Sigma$, which is a part of the output, and all the elements in $S_i$ are added to the set $\Xi$. At a high level, the goal of \Cref{alg:rounding} is to resolve all of the sets in either way, while maintaining the cost and the feasibility of the \LP.

\begin{algorithm}[H]
	\small
	\caption{RoundLP$(\X, \S, w, k, \sigma)$} \label{alg:rounding}
	\begin{algorithmic}[1]
		\State $\Sigma \gets \emptyset$,  $\Xi \gets \emptyset$
		\State $\Xc \gets \X$, $\Sc \gets \S$
		\While{$|\Sc| \ge 2$} \label{lin:outer-whileloop}
		\State $S_a \gets$ an arbitrary set from $\Sc$. \label{lin:Sa-chosen}
		\While{$0 < x_a < \alpha$ and $|\Sc \setminus \{S_a\}| \ge 1$} \label{lin:inner-whileloop}
		\State $S_b \gets$ an arbitrary set from $\Sc \setminus \{S_a\}$. \label{lin:Sb-chosen}
		
		\If{$\frac{|\Xc \cap S_a|}{w_a} \ge \frac{|\Xc \cap S_b|}{w_b}$} \label{lin:sets-compared}
		\State $(x_a, x_b, z) \gets $\Call{RoundTwoSets}{$S_a, S_b, w, \sigma, \Xc, \Sc$}
		\If{$x_b = 0$}
		\State $\Sc \gets \Sc \setminus \{S_b\}$
		\EndIf
		\If{$x_a = \alpha$}
		\State $\Xi \gets \Xi \cup S_a$, $\Xc \gets \Xc \setminus S_a$.
		\State $\Sigma \gets \Sigma \cup \{S_a\}, \Sc \gets \Sc \setminus \{S_a\}$
		\EndIf
		\Else
		\State $(x_b, x_a, z) \gets $\Call{RoundTwoSets}{$S_b, S_a, w, \sigma, \Xc, \Sc$}
		\If{$x_a = 0$}
		\State $\Sc \gets \Sc \setminus \{S_a\}$
		\State $S_a \gets S_b$ \label{lin:Sa-chosen-2}
		\EndIf
		\If{$x_b = \alpha$}
		\State $\Xi \gets \Xi \cup S_b$, $\Xc \gets \Xc \setminus S_b$.
		\State $\Sigma \gets \Sigma \cup \{S_b\}, \Sc \gets \Sc \setminus \{S_b\}$
		\EndIf
		\EndIf
		\EndWhile
		\EndWhile \label{lin:end-outerwhile}
		\State $\S_e \gets \S_c$ \label{lin:last-set}
		\State \Return $\Sigma \cup \S_e$ \vspace{2mm}
		\Statex \hrule
		\Function{RoundTwoSets}{$S_1, S_2, w, \sigma, \Xc, \Sc$} \label{fn:round2sets}
		\State $\delta \gets \min\{\alpha - x_1, \frac{w_2}{w_1} \cdot x_2\}$
		\State $x_1 \gets x_1 + \delta$
		\State $x_2 \gets x_2 - \frac{w_1}{w_2} \cdot \delta$
		\State For all elements $e_j \in \Xc$, update $z_j \gets \sum_{i: e_j \in S_i} x_i$
		\State \Return $(x_1, x_2, z)$
		\EndFunction
	\end{algorithmic}
\end{algorithm}

Given the procedure {\sc{RoundTwoSets}}, we choose the pairs of sets to be rounded carefully. In \Cref{lin:Sa-chosen}, we pick a set $S_a \in \S_c$ arbitrarily, and then we pair it up with another set $S_b \in \S_c$ chosen arbitrarily in \Cref{lin:Sb-chosen}. To ensure that the constraint \Cref{constr:k-cover} is maintained, we carefully determine whether to increase $x_a$ and decrease $x_b$ in {\sc{RoundTwoSets}}, or vice versa. Thinking of $\frac{|\X_c \cap S_a|}{w_a}$, and $\frac{|\X_c \cap S_b|}{w_b}$ as the ``cost-effectiveness'' of the sets $S_a$ and $S_b$ respectively, we increase $x_a$ at the expense of $x_b$, if $S_a$ is more cost-effective than $S_b$ or vice versa.

Notice that if we paired up sets $S_a$ and $S_b$ arbitrarily and rounded using {\sc{RoundTwoSets}}, the feasibility of the \LP may not be maintained. In particular, we cannot ensure that for all elements $e_j \in \X_c$, $z_j \le 1$. To this end, we maintain the following two invariants: 
\begin{enumerate}
	\item Let $\X_o = \{e_j \in \X_c \mid z_j \ge \alpha \}$. During the execution of while loop of \Cref{lin:outer-whileloop}, the elements of $\X_o$ are contained in the set $S_a \in \S_c$, that is chosen in \Cref{lin:Sa-chosen} or \Cref{lin:Sa-chosen-2}. \label[invar]{invar:inv1}
	\item Fix any set $S_i \in \S_c \setminus \{S_a\}$. The $x_i$ value is unchanged since the beginning of the algorithm until the beginning of the current iteration of while loop of \Cref{lin:inner-whileloop}; the $x_i$ value can change in the current iteration only if $S_i$ is paired up with $S_a$.  \label[invar]{invar:inv2}
\end{enumerate}
In \Cref{lem:z-feasibility}, we show that these invariants imply that \Cref{constr:fractional-z} is maintained. 

The invariants are trivially true before the start of the while loop. Let $S_a \in \Sc$ be a set chosen in \Cref{lin:Sa-chosen}, or \Cref{lin:Sa-chosen-2}. During the while loop, we maintain the invariants by pairing up the $S_a$ with other arbitrary sets $S_b$, until $S_a$ is removed from $\S_c$ in one of the two ways; or until it is the last set remaining. It is easy to see that \Cref{invar:inv2} is maintained -- we argue about \Cref{invar:inv1} in the subsequent paragraphs. From the second invariant, we have that if the change in the $z_j$ value for any element $e_j \in \X_c$ is positive, then it due to the change in the $x_a$ value corresponding to $S_a$ (recall that when the $x_i$ value of a set $S_i \in \S_c$ increases to $\alpha$, all the elements contained in it are removed from $\X_c$).

Now we describe in detail how the first invariant is being maintained in the course of the algorithm. Consider the first case, i.e. in {\sc{RoundTwoSets}}, we increase $x_a$ and decrease $x_b$. If after this, $x_b$ becomes $0$, then we remove $S_b$ from $\S_c$. If, on the other hand, $x_a$ increases to $\alpha$, then all the elements in $\X_c \cap S_a$ are covered to an extent of at least $\alpha$, and so we remove $S_a$ from $\S_c$ and $S_a \cap \X_c$ from $\X_c$. In the first case, the set $\X_o$ continues to be a subset of $S_a$, while in the second case, it becomes empty. Thus, \Cref{invar:inv1} is maintained automatically in both cases.

In the second case, in {\sc{RoundTwoSets}}, $x_a$ is decreased and $x_b$ is increased. This case is a bit more complicated, because $z_j$ values of elements $e_j \in S_b$ are being increased by virtue of increase in $x_b$. Therefore, we need to explicitly maintain \Cref{invar:inv1}. If $x_b$ reaches $\alpha$, then $S_b$ is removed from $\S_c$ and all the elements covered by $S_b$ are removed from $\X_c$ (and thus the invariant is maintained). On the other hand, if $x_a$ reaches $0$, then the net change in the $z_j$ values for the elements $e_j \in S_a \setminus S_b$ is non-positive -- this follows from \Cref{invar:inv2}, as the $x_i$ values of the sets in $\S_c \setminus \{S_a, S_b\}$ are unchanged, and $x_a$ is now zero. Therefore, the set $\X_o \cap (S_a \setminus S_b)$ becomes empty. However, $\X_o \cap S_b$ may be non-empty because of the increase in $x_b$. Therefore, we rename $S_b$ as $S_a$, and continue pairing it up with other sets. Notice that we have maintained \Cref{invar:inv1} although the set $S_a$ has changed.

From the above discussion, we have the following result.

\begin{claim} \label{cl:z-excess-set}
	Throughout the execution of the while loop of \Cref{lin:outer-whileloop}, \Cref{invar:inv1,invar:inv2} are maintained.
\end{claim}

Finally, if at the end of while loop of \Cref{lin:outer-whileloop}, we set $\S_e$ to be $\S_c$, and add it to our solution. Note that at this point, $\S_c$ can be empty, or it may contain one set. We show that in either case, the resulting solution $\Sigma \cup \S_e$ covers at least $k$ elements.

\subsection{Analysis}
In this section, we analyze the behavior of \Cref{alg:rounding}. In the following lemma, we show that in each iteration, we make progress towards rounding while maintaining the cost of the \LP solution.
\begin{lemma} \label[lem]{lem:weight-maintained}
Let $\sigma = (x, z), \sigma' = (x', z')$ be the \LP solutions just before and after the execution of  {\sc{RoundTwoSets}}$(S_1, S_2, w, \sigma, \Xc, \Sc)$ for some sets $S_1, S_2 \in \Sc$ in some iteration of the algorithm, such that $\sigma$ is a feasible solution to the \LP. Then,
\begin{enumerate}
	\item $cost(\sigma) = cost(\sigma')$.
	\item $\sum_{e_j \in \Xc} z'_j \ge \sum_{e_j \in \Xc} z_j$.
	\item Either $x_1' = \alpha$ or $x_2' = 0$ (or both).
\end{enumerate}
\end{lemma}
\begin{proof}
	\begin{enumerate}
		\item Note that the $x_i$ variables corresponding to all the sets $S_i \notin \{S_1, S_2\}$ remain unchanged. The net change in the cost of the \LP solution is $$w_1 \cdot (x_1' - x_1) + w_2 \cdot (x_2' - x_2) = w_1 \cdot \delta - w_2 \cdot \lr{\frac{w_1}{w_2} \cdot \delta} = 0.$$
		
		\item Let $A = S_1 \cap \Xc$, and $B = S_2 \cap \Xc$. $z'_j = z_j$ for all elements $e_j \not\in A \cup B$, i.e. $z_j$ values are modified only for the elements $e_j \in A \cup B$. 
		
		For $|A|$ elements $e_j \in A$, $z_j$ value is increased by $\delta$ by virtue of increase in $x_1$. Similarly, for $|B|$ elements $e_{j'} \in B$, $z_{j'}$ value is decreased by $\frac{w_1}{w_2} \cdot \delta$. However by assumption, we have that $\frac{|A|}{w_1} \ge \frac{|B|}{w_2}$. Therefore, the net change in the sum of $z_j$ values is
		$$|A| \cdot \delta - |B| \cdot \lr{\frac{w_1}{w_2} \cdot \delta } \ge |A| \cdot \delta - \lr{\frac{|A|}{w_1} \cdot w_1} \cdot \delta \ge 0.$$ 
		
		\item The value of $\delta$ is chosen such that $\delta = \min \{ \alpha - x_1, \frac{w_2}{w_1} \cdot x_2\}$. If $\delta = \alpha - x_1 \le  \frac{w_2}{w_1} \cdot x_2$, then $x_1' = x_1 - (\alpha - x_1) = \alpha$, and $x_2' = x_2 - \frac{w_1}{w_2} \cdot (\alpha - x_1) \ge x_2 - x_2 = 0$. In the other case when $\delta = \frac{w_2}{w_1} \cdot x_2 < (\alpha - x_1)$, we have that $x_1' = x_1 + \frac{w_2}{w_1} \cdot x_2 < x_1 + (\alpha - x_1) = \alpha$, and $x_2' = x_2 - \frac{w_2}{w_1} \cdot \frac{w_1}{w_2} \cdot x_2 = 0$.
	\end{enumerate}
\end{proof}

\begin{remark}
	Note that \Cref{lem:weight-maintained} (in particular Part 2 of \Cref{lem:weight-maintained}) alone is not sufficient to show the feasibility of the \LP after an execution of {\sc{RoundTwoSets}}---we also have to show that $z_j' \le 1$. This is slightly involved, and is shown in \Cref{lem:z-feasibility} with the help of \Cref{invar:inv1,invar:inv2}.
\end{remark}

\begin{corollary}
	\Cref{alg:rounding} runs in polynomial time.
\end{corollary}
\begin{proof}
	In each iteration of the inner while loop \Cref{lin:inner-whileloop}, {\sc{RoundTwoSets}} is called on some two sets $S_1, S_2 \in \Sc$, and as such from \Cref{lem:weight-maintained}, either $x_1' = \alpha$ or $x_2' = 0$. Therefore, at least one of the sets is removed from $\Sc$ in each iteration. Therefore, there are at most $O(|\S|)$ iterations of the inner while loop. It is easy to see that each execution of {\sc{RoundTwoSets}} takes $O(|\S| \cdot |\X|)$ time.
\end{proof}

In the following Lemma, we show that \Cref{constr:fractional-z} is being maintained by the algorithm. This, when combined with \Cref{lem:weight-maintained}, shows that we maintain the feasibility of the \LP at all times.
\begin{lemma} \label[lem]{lem:z-feasibility}
	During the execution of \Cref{alg:rounding}, for any element $e_j \in \Xc$, we have that $z_j \le 2 \alpha$. By the choice of range of $\alpha$, the feasibility of the \LP is maintained.
\end{lemma}
\begin{proof}
	At the beginning of the algorithm, we have that $z_j \le \alpha$ for all elements $e_j \in \X_c = \X$. Now at any point in the while loop, consider the set $\X_o = \{e_j \in \Xc \mid  z_j > \alpha\}$ as defined earlier. For any element $e_j \in \Xc \setminus \X_o$, the condition is already met, therefore we need to argue only for the elements in $\X_o$. We know by \Cref{invar:inv1} that there exists a set $S_a \in \S_c$ such that $\X_o \subseteq S_a$.
	
	By \Cref{invar:inv2}, the $x_i$ values of all sets $S_i \in \S_c \setminus \{S_a\}$ are unchanged, and therefore for all elements $e_j \in \X_c$, the net change to the $z_j$ variable is positive only by the virtue of increase in the $x_a$ value. However, the net increase in the $x_a$ value is at most $\alpha$ because a set $S_i$ is removed from $\S_c$ as soon as its $x_i$ value reaches $\alpha$. Accounting for the initial $z_j$ value which is at most $\alpha$, we conclude that $z_j \le 2 \alpha$.
\end{proof}

Note that after the end of while loop (\Cref{lin:end-outerwhile}), we must have $|\S_c| \le 1$. That is in \Cref{lin:last-set}, we either let $\S_e \gets \S_c = \emptyset$, or $\S_e \gets \S_c = \{S_i\}$ for some set $S_i \in \S$.

To state the following claim, we introduce the following notation. Let $\sigma' = (x', z')$ be the \LP solution at the end of \Cref{alg:rounding}. Let $\S_r = \S \setminus \Sigma$, where $\Sigma$ is the collection at the end of the while loop of \Cref{alg:rounding}, and let $\X_r = \X \setminus \Xi$. Note that any element $e_j \in \X_r$ is contained only in the sets of $\S_r$. Finally, let $Z_r = \sum_{e_j \in \X_r} z'_j$.
\begin{claim} \label{cl:last-set}
	If $\S_e \neq \emptyset$, then at least $Z_r$ elements are covered by $\S_e$.
\end{claim}
\begin{proof}
	 By assumption, we have that $\S_e \neq \emptyset$, i.e. $\S_e = \{S_i\}$ for some $S_i \in \S$. For each $S_l \in \S_r$ with $l \neq i$, we have that $x_l = 0$, again by the condition of the outer while loop. Since \Cref{constr:coverage} is made tight for all elements in each execution of {\sc{RoundTwoSets}}, for any element $e_{j'} \in \X_r$ but $e_{j'} \not\in S_i$, we have that $z_{j'} = 0$. On the other hand, for elements $e_j \in \X_r \cap S_i$, we have that $z'_j = x'_i \le \alpha$. If the number of such elements is $p$, then we have that $Z_r \le \alpha \cdot p$. The lemma follows since choosing $S_i$ covers all of these $p$ elements, and $p \ge Z_r/\alpha \ge Z_r.$
\end{proof}

In the following lemma, we show that \Cref{alg:rounding} produces a feasible solution.
\begin{lemma} \label[lem]{lem:atleast-k}
	The solution $\Sigma \cup \S_e$ returned by \Cref{alg:rounding} covers at least $k$ elements.
\end{lemma}
\begin{proof}
	There are two cases -- $\S_e = \emptyset$, or $\S_e = \{S_i\}$ for some $S_i \in \S$. In the first case, all elements in $\X_r$ are uncovered, and for all such elements, $e_{j'} \in \X_r$, we have that $z_{j'} = 0$. In this case, it is trivially true that the number of elements of $\X_r$ covered by $\S_e$ is $Z_r$ $(= 0)$. In the second case, the same follows from \Cref{cl:last-set}. Therefore, in both cases we have that, 
	\begin{align*}
		\text{Number of elements covered} &\ge |\Xi| + Z_r
		\\&\ge \sum_{e_j \in \Xi} z'_j + \sum_{e_j \in \X_c} z'_i \tag{By \Cref{lem:z-feasibility} and $z_j' \le 1$}
		\\&= \sum_{e_j \in \X} z'_j
		\\&\ge \sum_{e_j \in \X} z_j \tag{\Cref{lem:weight-maintained}, Part 2}
		\\&\ge k \tag{By \Cref{lem:feasibility} and \Cref{constr:k-cover}}
	\end{align*}
	Recall that $z_j$ refers to the $z$-value of an element $e_j$ in the optimal LP solution $\sigma$, at the beginning of the algorithm.
\end{proof}

\begin{lemma} \label[lem]{lem:mainlemma}
	Let $\Sigma \cup \S_e$ be the solution returned by \Cref{alg:rounding}, and let $B$ be the weight of the heaviest set in $S$. Then, 
	\begin{enumerate}
		\item $w(\Sigma) \le \frac{1}{\alpha} \sum_{S_i \in \Sigma} w_i x'_i$
		\item $w(\S_e) \le B$ 
		\item $w(\Sigma \cup S_e) \le \frac{1}{\alpha} \sum_{S_i \in \Sigma} w_i x'_i + B$
	\end{enumerate}
\end{lemma}
\begin{proof}
	For the first part, note that a set $S_i$ is added to $\Sigma$ only if $x'_i \ge \alpha$. For the second part, note that $\S_e$ contains at most one set $S_i \in \S$. By definition, weight of any set in $S_i$ is bounded by $B$, the maximum weight of any set in $\S$. The third part follows from the first and the second parts.
\end{proof}

From \Cref{lem:atleast-k} and \Cref{lem:mainlemma}, we conclude that $\Sigma \cup \S_e$ is the solution that covers at least $k = k_p - |\X_1|$ elements from $\X_p \setminus \X_1$, and whose cost is at most $\frac{1}{\alpha} \sum_{S_i \in \Sigma} w_i x'_i + B$.

\section{Generalization of \PSC}
Consider the following generalization of \PSC problem, where the elements $e_j \in \X'$ have profits $p_j \ge 0$ associated with them. Now the goal is to choose a minimum-weight collection $\Sigma \subseteq \S'$ such that the total profit of elements covered by the sets of $\Sigma$ is at least $K$, where $0 \le K \le \sum_{e_j \in \X} p_j$ is provided as an input. Note that setting $p_j = 1$ for all elements we get the original \PSC problem. This generalization has been considered in \cite{KonemannPS2011}.

It is easy to modify our algorithm that for \PSC, such that it returns a $2\beta + 2$ approximate solution for this generalization as well. We briefly describe the modifications required. Firstly, we modify \Cref{constr:k-cover} of \PSC \LP to incorporate the profits as follows: $$\sum_{e_j \in \X} z_j \cdot p_j \ge K$$ The preprocessing and the rounding algorithms work with the straightforward modifications required to handle the profits. One significant change is in the rounding algorithm (\Cref{alg:rounding}). We compare the ``cost-effectiveness'' of the two sets $S_a, S_b$ in \Cref{lin:sets-compared} for the \PSC as $\frac{|S_a \cap \X_c|}{w_a} \ge \frac{|S_b \cap \X_c|}{w_b}$. For handling the profits of the elements, we replace this with the following condition --  $\frac{P_a}{w_a} \ge \frac{P_b}{w_b}$, where $P_a := \sum_{e_j \in S_a \cap \X_c} p_j$, and $P_b := \sum_{e_j \in S_b \cap \X_c} p_j$. With similar straightforward modifications, the analysis of \Cref{alg:rounding} goes through with the same guarantee on the cost of the solution. We remark here that despite the profits, the approximation ratio only depends on that of the standard \SC \LP, which is oblivious to the profits.
\bibliography{wpgsc}
\end{document}